\documentclass[nofootinbib,pra,reprint,amsmath,amssymb,aps,superscriptaddress]{revtex4-2}

\usepackage{tikz}

\usepackage{graphicx} 

\usepackage{amsthm}
\newtheorem{thm}{Theorem}

\usepackage{mathrsfs}

\usepackage{amsfonts}

\usepackage[hidelinks]{hyperref}
\usepackage{orcidlink}

\usepackage{cleveref}
\crefname{section}{\S}{\SS}
\crefname{appendix}{Appendix }{\SS}
\crefformat{thm}{\mbox{#2Thm. #1#3}}
\crefformat{figure}{#2Fig. #1#3}
\crefformat{table}{#2Table #1#3}
\crefformat{equation}{#2Eq. (#1)#3}

\newcommand{\tr}[1]{\mathrm{tr}\left\{ #1\right\}}

\newcommand{\dd}{\mathrm{d}}

\begin{document}

\author{Ruvi Lecamwasam\,\orcidlink{0000-0001-6531-3233}}
\email{me@ruvi.blog}
\affiliation{A*STAR Quantum Innovation Centre (Q.Inc), Institute of Materials Research and Engineering (IMRE), Agency for Science, Technology and Research (A*STAR), 2 Fusionopolis Way, 08-03 Innovis 138634, Singapore \looseness=-1}
\affiliation{Quantum Machines Unit, Okinawa Institute of Science and Technology Graduate University, Onna, Okinawa 904-0495, Japan \looseness=-1}
\affiliation{Centre for Quantum Computation and Communication Technology, Department of Quantum Science and Technology, Australian National University, ACT 2601, Australia \looseness=-1}
\author{Syed Assad\,\orcidlink{0000-0002-5416-7098}}
\affiliation{A*STAR Quantum Innovation Centre (Q.Inc), Institute of Materials Research and Engineering (IMRE), Agency for Science, Technology and Research (A*STAR), 2 Fusionopolis Way, 08-03 Innovis 138634, Singapore \looseness=-1}
\affiliation{Centre for Quantum Computation and Communication Technology, Department of Quantum Science and Technology, Australian National University, ACT 2601, Australia \looseness=-1}
\author{Joseph J. Hope\,\orcidlink{0000-0002-5260-1380}}
\affiliation{Department of Quantum Science and Technology, Australian National University, ACT 2601, Australia \looseness=-1}
\author{\\Ping Koy Lam\,\orcidlink{0000-0002-4421-601X}}
\affiliation{A*STAR Quantum Innovation Centre (Q.Inc), Institute of Materials Research and Engineering (IMRE), Agency for Science, Technology and Research (A*STAR), 2 Fusionopolis Way, 08-03 Innovis 138634, Singapore \looseness=-1}
\affiliation{Centre for Quantum Technologies, National University of Singapore, 3 Science Drive 2, Singapore 117543 \looseness=-2}
\affiliation{Centre for Quantum Computation and Communication Technology, Department of Quantum Science and Technology, Australian National University, ACT 2601, Australia \looseness=-1}
\author{Jayne Thompson\,\orcidlink{0000-0002-3746-244X}}
\email{thompson.jayne2@gmail.com}
\affiliation{Horizon Quantum Computing, 05-22 Alice@Mediapolis, 29 Media Circle, Singapore 138565 \looseness=-1}
\affiliation{Institute of High Performance Computing, Agency for Science, Technology, and Research (A*STAR), Singapore 138634, Singapore \looseness=-2}
\author{Mile Gu\,\orcidlink{0000-0002-5459-4313}}
\email{mgu@quantumcomplexity.org}
\affiliation{Nanyang Quantum Hub, School of Physical and Mathematical Sciences, Nanyang Technological University,  21 Nanyang Link, Singapore 639673 \looseness=-2}
\affiliation{Centre for Quantum Technologies, National University of Singapore, 3 Science Drive 2, Singapore 117543 \looseness=-2}
\affiliation{MajuLab, CNRS-UNS-NUS-NTU International Joint Research Unit, UMI 3654, 117543, Singapore\looseness=-2}

\date{\today}

\title{The relative entropy of coherence quantifies performance in Bayesian metrology}

\begin{abstract}
  The ability of quantum states to be in superposition is one of the key features that sets them apart from the classical world. This `coherence' is rigorously quantified by resource theories, which aim to understand how such properties may be exploited in quantum technologies. 
  There has been much research on what the resource theory of coherence can reveal about quantum metrology, almost all of which has been from the viewpoint of Fisher information. We prove however that the relative entropy of coherence, and its recent generalisation to POVMs, naturally quantify the performance of Bayesian metrology. 
  In particular, we show how a coherence measure can be applied to an ensemble of states.
  We then prove that during parameter estimation, the ensemble relative entropy of coherence (C) is equal to the difference between optimal Holevo information (X), and the mutual information attained by the measurement (I). We call this relation the CXI equality. 
  The ensemble coherence lets us visualise how much information is locked away in superposition inaccessible with a given measurement scheme, and quantify the advantage that would be gained by using a joint measurement on multiple states.
  Our results hold regardless of how the parameter is encoded in the state, encompassing unitary, dissipative, and discrete settings. We consider both projective measurements, and general POVMs.
  This work suggests new directions for research in coherence, provides a novel operational interpretation for the relative entropy of coherence and its POVM generalisation, and introduces a new tool to study the role of quantum features in metrology.
\end{abstract}

\maketitle

\section{Introduction}\label{sec:Introduction}
Superposition is one of the most fundamental and unique aspects of quantum physics. This is rigorously quantified by the resource theory of coherence \cite{baumgratz_quantifying_2014,streltsov_colloquium_2017,chitambar_quantum_2019}, which studies how the amount of superposition over a given basis relates to physical properties of a state and its ability to perform useful tasks. In recent years there has been intensive investigation into the role of coherence in areas including thermodynamics \cite{lostaglio_quantum_2015,korzekwa_extraction_2016,santos_role_2019}, quantum information processing \cite{hillery_coherence_2016,shi_coherence_2017,contreras-tejada_resource_2019,ahnefeld_coherence_2022}, and quantum correlations such as entanglement \cite{streltsov_measuring_2015,ma_converting_2016,jin_quantifying_2018,tan_quantifying_2017}. 

Classical systems can only exist as a statistical mixture of orthogonal states --- such mixtures are termed \emph{incoherent}. In resource theories, we quantify the \emph{coherence} of a quantum system by measuring its distance from the set of incoherent states. There are many different notions of distance, leading to a variety of coherence measures \cite{streltsov_colloquium_2017}. The oldest, and one of the simplest, is the relative entropy of coherence, defined using the quantum relative entropy as the distance measure \cite{baumgratz_quantifying_2014}. It is an outstanding question which coherence measures are more fundamental, or relevant to different problems \cite{adesso_measures_2016,streltsov_colloquium_2017}. Answering this question requires finding applications and operational interpretations for these coherence measures.

One such application of great interest to quantum technologies is metrology \cite{degen_quantum_2017,polino_photonic_2020,demkowicz-dobrzanski_quantum_2015}. It has long been known that quantum mechanics can be used to measure physical parameters more efficiently than is classically possible \cite{giovannetti_quantum_2006,giovannetti_advances_2011}. Precisely which quantum properties allow for this, and how they may best be exploited, is an area of active research \cite{giovannetti_quantum_2006,degen_quantum_2017}. Depending on how a quantum system is measured, some information is usually `locked away', inaccessible to the measurement scheme \cite{salmon_only_2023}. It is intuitive that this inaccessible information has something to do with superposition. For example, projective measurement in some basis is insensitive to phase information between the basis coefficients. It is thus natural to ask what insight the resource theory of coherence can provide in metrology \cite{marvian_how_2016,tan_coherence_2018,baek_quantifying_2020}. 

Broadly speaking there are two lenses through which parameter estimation may be viewed. The widely used \emph{Fisher information} quantifies the smallest fluctuation in the parameter which may be detected around some operating point \cite{demkowicz-dobrzanski_quantum_2015}. Often however we do not know the parameter accurately, but instead have a probability distribution over a large range of values. Moreover the unrealistic assumption of an operating point known with infinite precision can sometimes imply performance which is not attainable in practice \cite{hall_universality_2012,demkowicz-dobrzanski_multi-parameter_2020}. Fisher information also requires continuous evolution of the system with respect to the parameter, which excludes cases such as state discrimination. These issues are addressed by the Bayesian approach \cite{demkowicz-dobrzanski_quantum_2015,demkowicz-dobrzanski_multi-parameter_2020,fuchs_priors_2009}, which instead quantifies the average error in our estimate of the parameter. Fisher information can generally be computed using derivatives of the quantum state, which naturally capture the change due to an infinitesimal fluctuation in the parameter. Analysing Bayesian metrology however is generally more complicated \cite{wiebe_efficient_2016,macieszczak_bayesian_2014,paesani_experimental_2017,ruster_entanglement-based_2017,berry_how_2009}, as we must consider all possible values of the parameter, and all possible measurement results in each case. Quantum information theory is increasingly proving itself to be a useful tool for this \cite{hall_does_2012,nair_fundamental_2018,hall_better_2022}.

The majority of research into coherence and metrology has focused on Fisher information. Detailed studies have shown that the quantum Fisher information is related to the second derivative of the relative entropy of coherence, though there are additional factors in the relation that are still not well understood \cite{giorda_coherence_2017}. Other authors have also found complex relationships between the two \cite{yu_effects_2020,wang_coherence_2018}. The relationship between quantum Fisher information and coherence in general dynamics has been studied in \cite{ye_quantum_2018,jafari_dynamics_2020,mohamed_two-qubit_2021}, and there have even been formulations which define coherence using Fisher information \cite{feng_quantifying_2017,li_quantum_2021,tan_fisher_2021}.

In contrast to the wealth of research investigating the connection between coherence and Fisher information, the role of coherence in Bayesian metrology has received relatively little attention. In \cite{ares_signal_2021,alvarez-marcos_phase-space_2022} the authors studied unitary or displacement encoding of the parameter, and a particular detection scheme analogous to measuring the phase evolution along a `ruler'. They defined a notion of coherence analogous to the classical Wiener-Khintchine theorem, and related this to the measurement resolution. Another example is \cite{hall_does_2012}, where the relative entropy of coherence (termed `$G$-assymetry' by the authors) was used to study the limits of nonlinear metrology with unitary parameter encodings.

In this work we show that the relative entropy of coherence \cite{baumgratz_quantifying_2014} and its recent generalisation to POVMs \cite{bischof_resource_2019,bischof_quantifying_2021} are fundamentally related to Bayesian metrology. We assume nothing about how the parameter is encoded in the state, even allowing for situations where the encoding is not continuous. Our results apply to both projective measurement and general POVMs.
In this general case, we show that the relative entropy of coherence naturally emerges as the quantifier of information gain from the measurement.

The manuscript is structured as follows. We begin in \cref{sec:Background} with an overview of some necessary concepts from Bayesian estimation and information theory. In \cref{sec:CXI} we show how the relative entropy of coherence may be applied to an ensemble of states. We then prove that the ensemble coherence (C) is equal to the difference between the optimal Holevo information assuming infinite resources (X), and the mutual information attained by a particular projective measurement (I), a relation we call the \emph{CXI equality}. In \cref{sec:CXIPOVM} we generalise the CXI equality to POVM measurements. We then apply this in \cref{sec:Measurement} to a simple example of state discrimination on the Bloch sphere, the code for which is available at \cite{lecamwasam_cxi_2023}. Finally in \cref{sec:Conclusion} we discuss promising directions for future research.



\section{Information theory and Bayesian metrology}\label{sec:Background}
Quantum metrology exploits properties such as interference and entanglement for high-precision sensing. There is a veritable zoo of possible schemes, but the majority follow the general procedure of initialising a probe state $\rho_0$, interacting this with a parameter $\phi$ of interest resulting in a state $\rho_{\phi}$, and then measuring $\rho_{\phi}$ to gain information about the parameter \cite{degen_quantum_2017}. Typically many measurements are performed \cite{denot_adaptive_2006,teklu_bayesian_2009,berry_how_2009}, and the results then processed to make an estimate of the parameter \cite{trees_detection_2013}. Much work has been done to optimise each of these steps \cite{demkowicz-dobrzanski_quantum_2015,berry_how_2009}. Finding fundamental bounds on performance is thus a complex task \cite{holevo_probabilistic_2011,helstrom_quantum_1969,demkowicz-dobrzanski_quantum_2015,degen_quantum_2017}, but in recent years quantum information theory \cite{wilde_quantum_2017} has shown itself to be a powerful tool for addressing this problem \cite{hall_does_2012,hall_better_2022,nair_fundamental_2018}.

Here we will follow standard probability conventions \cite{resnick_probability_2014} and use a capital letter such as $\Phi$ to represent a general probabilistic event, and a lower case letter $\phi$ for a specific outcome. For example, $\Phi$ may be the outcome of a dice roll, which can take values $\phi=1$ through to $\phi=6$. Each outcome has an associated probability $p_{\Phi}(\phi)$. In information theory, the uncertainty in the value of $\Phi$ is represented by the entropy:
\begin{equation}
  H({\Phi})=-\sum_{\phi}p_{\Phi}(\phi)\log p_{\Phi}(\phi).
\end{equation}
Entropy has a geometric interpretation in that when exponentiated it gives the volume of possible values that $\Phi$ may take \cite[\S 3]{cover_elements_2005}. For example for a fair dice $p_\Phi(\phi)=1/6$ so $H(\Phi)=\log 6$, representing $e^{\log 6}=6$ possible outcomes. However, if you learn via measurement that the dice's value was $3$, the entropy becomes $0$ representing a single value $e^0=1$. If the outcome were known to be either $2$ or $5$ with equal probability, the entropy would be $\log 2$ corresponding to two possible outcomes. The case of continuous $\Phi$ is analogous with sums replaced by integrals, and the exponentiated entropy then represents the volume of possible values \cite{cover_elements_2005}.


Suppose $\Phi$ represents an unknown parameter that we wish to estimate, such as the strength of a magnetic field. In Bayesian metrology, the probability distribution $p_{\Phi}$ then describes our initial knowledge about the value of $\Phi$. For example, for physical reasons we may know it to be Gaussian distributed about some $\phi_0$. To apply the Bayesian framework in the case of no prior knowledge, one can take $p_{\Phi}$ to be the probability distribution with maximum entropy that satisfies physical constraints, which is typically the uniform distribution.

We begin our estimation process by interacting a probe state $\rho_0$ with $\Phi$. This results in an ensemble $\mathcal{E}_{\Phi}$ of states $\rho_{\phi}$ with probability $p_{\Phi}(\phi)$. 
The ensemble state is then the average over the possible parameter values:
\begin{equation}\label{eq:EnsembleState}
  \rho_{\Phi}=\sum_{\phi}p_{\Phi}(\phi)\rho_{\phi}.
\end{equation}
We choose a basis and perform a measurement on $\rho_{\Phi}$, whose outcome is a random value $M$. After observing a measurement result, we can use Bayes' theorem to combine this with our prior information to construct the posterior probability distribution conditioned on our measurement $p_{\Phi|M}$. Note that $M$ can represent either a single measurement, or a sequence of measurement results in different bases.
Finally we make an estimate $\hat{\phi}$ of $\phi$ which is some function of the final distribution $p_{\Phi|M}$, such as the mean or median \cite{demkowicz-dobrzanski_quantum_2015,trees_detection_2013}. 

To quantify the performance of Bayesian metrology, we look at the expected value of $(\phi-\hat{\phi})^2$. This is called the Average Mean Squared Error (AMSE). Computing the AMSE can be quite challenging, since we must take into account the possible values of $\phi$, all possible sequences of measurement results, and whatever algorithm is used to extract $\hat{\phi}$ from the final probability distribution. We will see that we can bound this error using information theory.

The entropy of the final distribution $p_{\Phi|M}$ is denoted $H(\Phi|M)$. This will necessarily be less than or equal to $H(\Phi)$ \cite{cover_elements_2005,wilde_quantum_2017}, since measurement can only reduce uncertainty in the parameter. This decrease is quantified by the mutual information
\begin{equation}
  I(\Phi;M)=H(\Phi)-H(\Phi|M),
\end{equation}
where the `$;$' is to emphasise that $I$ is symmetric in its arguments \cite{cover_elements_2005}. The mutual information provides a lower bound on the average mean squared error \cite{hall_does_2012}, and repeatedly maximising mutual information optimises performance for multi-round measurement schemes \cite{davies_information_1978}.

Different measurement schemes will provide more or less information about the parameter. A natural question to ask is, given a particular ensemble $\mathcal{E}_{\Phi}$, what is the maximum mutual information that could be obtained from the optimum choice of measurement. The answer to this is provided by the Holevo information \cite{wilde_quantum_2017}. The entropy of a quantum state $\rho$, termed the von Neumann entropy, is defined as
\begin{equation}
  S(\rho)=-\mathrm{tr}\{\rho\log\rho\}.
\end{equation}
The Holevo information of the ensemble $\mathcal{E}_{\Phi}$ is then
\begin{equation}\label{eq:HolevoInformation}
  \chi(\mathcal{E}_{\Phi})=S\left(\rho_{\Phi}\right)-\sum_{\phi}p_{\Phi}(\phi)S(\rho_{\phi}).
\end{equation}
The Holevo information provides an upper bound on the mutual information: $I(M;\Phi)\le\chi(\mathcal{E})$. It thus quantifies the maximum information that can be extracted per probe state. However in general to attain the Holevo information, we must obtain $N$ identical probes $\rho_{\phi}$, apply an entangling unitary operation, and then perform the optimum multi-partite measurement. In the limit $N\rightarrow\infty$ the information gained approaches $N\chi(\mathcal{E})$, so the Holevo information is the information-per-probe assuming infinite resources and technology.

It is thus natural to ask about the difference between the optimum Holevo information $\chi$, and the mutual information $I$ actually attained from a given measurement scheme. In the following section we will show that this is given by the relative entropy of coherence \cite{baumgratz_quantifying_2014}, and its generalisation to POVMs \cite{bischof_resource_2019}.



\section{Coherence and Bayesian Metrology}\label{sec:CXI}

In this section we will derive the CXI equality for the case of projective measurement, which shows that the relative entropy of coherence quantifies the difference between the Holevo and mutual informations. In our derivations we will take the parameter to be discrete, the continuous case is analogous with sums replaced by integrals. We will denote our parameter as $\Phi$, and let $\mathcal{E}_{\Phi}$ denote the ensemble of states $\left\{\left(\rho_{\phi},p_{\Phi}(\phi)\right)\right\}$. Our measurement result will be $M$, and we will refer to the basis of the projective measurement as the `basis of $M$'. 

We wish to study the difference $\chi(\mathcal{E}_{\Phi})-I(\Phi;M)$. The information lost by a projective measurement is related to the `amount of superposition' of the probe $\rho_{\phi}$ over the basis of $M$. If $\rho_{\phi}$ is a measurement basis state, then the measurement will return the exact state with no information loss. However if the probe is in a superposition of many basis states, then a measurement can only return part of the information in the state. We will need to measure multiple times to recover the weights of the superposition, and choose different bases to gain phase information. Thus it seems intuitive there should be a relationship between the mutual information $I(\Phi;M)$, and the coherence of the ensemble states in the basis of $M$.

In this work we will use the relative entropy of coherence -- referred to from now on as the \emph{coherence} -- the simplest and most widely used measure \cite{streltsov_colloquium_2017}. Let the measurement $M$ be projection onto an orthogonal basis $\{\lvert m\rangle\}$. For a quantum state $\rho$, let $\Delta_M[\rho]$ represent the same state decohered in the basis of $M$:
\begin{equation}
  \Delta_M[\rho]=\sum_m\langle m|\rho |m\rangle\,|m\rangle\langle m|.
\end{equation}
In other words, $\Delta_M[\rho]$ sets off-diagonal elements to zero in the matrix of $\rho$ in the basis of $M$.  Then the coherence of $\rho$ with respect to this basis is defined as
\begin{equation}\label{eq:Coherence}
  C_M(\rho)=S\left(\Delta_M[\rho]\right)-S(\rho),
\end{equation}
where $S(\rho)=-\tr{\rho\log\rho}$ is the von Neumann entropy. 
It can be shown that $C_M(\rho)$ is equal to the quantum relative entropy between $\rho$ and $\Delta_M[\rho]$ \cite[\S III.C.1]{streltsov_colloquium_2017}: 
\begin{equation}
  C_M(\rho)=\mathrm{tr}\left\{\rho\log\rho-\rho\log\Delta_M[\rho]\right\}.
\end{equation}

Let us consider the information that is lost when we projectively measure an ensemble $\mathcal{E}_{\Phi}$. Our first thought may be to look at the average information loss upon measurement: $\sum_{\phi}p_{\Phi}(\phi)C_M(\rho_{\phi})$. However this represents loss of information, both about the parameter, and also the quantum state itself. It is only the former that is of interest in parameter estimation, thus we must subtract the coherence of the ensemble state. This leads us to define the \emph{ensemble coherence} of $\mathcal{E}_{\Phi}$ as
\begin{equation}\label{eq:EnsembleCoherence}
  \begin{aligned}
    C_M(\mathcal{E}_{\Phi})
    &=\sum_{\phi}p_{\Phi}(\phi)C_M(\rho_{\phi})-C_M\left(\sum_{\phi}p_{\Phi}(\phi)\rho_{_\phi}\right), \\
    &= \sum_{\phi}p_{\Phi}(\phi)C_M(\rho_{\phi})-C_M\left(\rho_{\Phi}\right).
  \end{aligned}
\end{equation}
We note that is analogous to the definiton of the Holevo information in \cref{eq:HolevoInformation}, replacing entropy with coherence. Coherence decreases under the mixing of quantum states \cite[\S III.C.1]{streltsov_colloquium_2017}, thus $C_M(\mathcal{E}_{\Phi})$ is always positive. To the best of our knowledge, this notion of the coherence of an ensemble is novel.

Our primary result is to show that the ensemble coherence is equal to the difference between the Holevo and mutual informations:
\begin{thm}[CXI equality for projective measurements]\label{thm:CXI}
  Let $\Phi$ be a parameter with probability distribution $p_{\Phi}(\phi)$, with a corresponding ensemble $\mathcal{E}_{\Phi}$ of states $\rho_{\phi}$. If we perform a single projective measurement on the ensemble, the CXI equality holds:
  \begin{equation}\label{eq:aCXI}
    C_M(\mathcal{E}_{\Phi})=\chi(\mathcal{E}_{\Phi})-I(\Phi;M),
  \end{equation}
  where $C_M(\mathcal{E}_{\Phi})$ is the ensemble coherence as defined in \cref{eq:EnsembleCoherence}, and the coherence measure is the relative entropy of coherence.
\end{thm}

This does not assume any particular form of the interaction between the probe and parameter which generates the states $\rho_{\phi}$. Thus \cref{eq:aCXI} holds for unitary encodings such as $\rho_{\phi}=e^{-iG\phi}\rho_0e^{iG\phi}$ for some Hermitian generator $G$, as well as dissipative evolution, and discontinuous settings where a separate $\rho_{\phi}$ is specified for each value of $\phi$. It also applies to both continuous and discrete parameters, and for quantum states in finite and infinite dimensions. As we will discuss at the end of the section, the latter requires the additional condition that the prior distribution has finite entropy: $H(\Phi)<\infty$, which is satisfied by all physical distributions \cite{baccetti_infinite_2013}.

Let us now prove \cref{thm:CXI}. We estimate a parameter $\phi$ encoded in a probe state $\rho_{\phi}$ by making a projective measurement $M$. Let $M$ have eigenbasis $\{|m\rangle\}$ with corresponding orthonormal projectors $\Pi_m$. We will abbreviate the probability $p_{M}(j)$ of observing the $j$th result as $p_j$:
\begin{equation}
   p_j=p_M(j).
\end{equation}
Then if $\rho$ is the state being measured, we have $p_m=\tr{\Pi_m\rho}$. The entropy of $M$ is
\begin{equation}
  H(M)=-\sum_mp_m\log p_m.
\end{equation}
We will first show that this is equal to the entropy of $\rho$ decohered in the orthogonal basis of $M$:
\begin{equation}
  H(M)=S(\Delta_M[\rho]).
\end{equation}
To see this we expand the right hand side as 
\begin{equation}
  S(\Delta_M[\rho])=-\tr{\Delta_M[\rho]\log\Delta_M[\rho]}.
\end{equation}
The decohered state $\Delta_M[\rho]$ is a diagonal matrix in the basis of $M$, where the elements of the diagonal are the probabilities $p_m$ of measurement outcome $m$. The entropy is then
\begin{equation}\label{eq:aEntropyOfM}
  \begin{aligned}
    S(\Delta_M[\rho])&=-\tr{
      \left(\begin{array}{ccc}
        p_1 &   &  \\
            & \ddots & \\
            & & p_{N_M}
      \end{array}\right)
      \log\left(\begin{array}{ccc}
        p_1 &   &  \\
            & \ddots & \\
            & & p_{N_M}
      \end{array}\right)
    }, \\
    &=-\tr{
      \left(\begin{array}{ccc}
        p_1\log p_1 &   &  \\
            & \ddots & \\
            & & p_{N_M}\log p_{N_M}
      \end{array}\right)
    }, 
  \end{aligned}
\end{equation}
which is equal to $H(M)$. Note that in \cref{eq:aEntropyOfM} off-diagonal elements of the matrices are zero.

We are now prepared to prove \cref{thm:CXI}. 
\begin{proof}[Proof of CXI equality]
  We begin by expanding the right hand side of \cref{eq:aCXI} in terms of the entropies of quantum states. The Holevo information is defined as
  \begin{equation}
    \chi(\mathcal{E}_{\Phi})=S(\rho_{\Phi})-\sum_{\phi}p_{\Phi}(\phi) S(\rho_\phi),
  \end{equation}
  where $\rho_{\Phi}$ is the ensemble state defined in \cref{eq:EnsembleState}. To calculate the mutual information we will use the expression $I(\Phi;M)=H(M)-H(M|\Phi)$ \cite{wilde_quantum_2017}. We showed earlier that $H(M)=S(\Delta_M[\rho_{\Phi}])$, while for the conditional entropy:
  \begin{equation}
    \begin{aligned}
      H(M|\Phi) &= \sum_{\phi}p_{\Phi}(\phi) H(M|\phi), \\
        &= \sum_{\phi}p_{\Phi}(\phi) S(\Delta_M[\rho_\phi]).
    \end{aligned}
  \end{equation}
  The right hand side of \cref{eq:aCXI} is thus
  \begin{equation}
    \begin{aligned}
      &\left(S(\rho_{\Phi})-\sum_{\phi}p_{\Phi}(\phi) S(\rho_\phi)\right) \\
      &\phantom{=}-\left(S(\Delta_M[\rho_{\Phi}])-\sum_{\phi}p_{\Phi}(\phi) S(\Delta_M[\rho_\phi])\right) \\
      &=\sum_{\phi}p_{\Phi}(\phi)\left(S(\Delta_M[\rho_\phi])-S(\rho_\phi)\right) \\
      &\phantom{=}-\left(S(\Delta_M[\rho_{\Phi}])-S(\rho_{\Phi})\right), \\
      &= C_M(\mathcal{E}_{\Phi}),
    \end{aligned}
  \end{equation}
  where in the last line we recalled the definition of the ensemble coherence \cref{eq:EnsembleCoherence}.
\end{proof}

Let us discuss the validity of our proof in infinite-dimesional Hilbert spaces. In finite dimensions entropy is always finite, thus all of the sums above converge, and we do not have to worry about expressions such as $\infty-\infty$. However, many common quantum systems require an infinite-dimensional Hilbert space, optics being a prominent example. Let us first consider the parameter $\phi$. Continuous parameters, and discrete parameters with an infinite number of values, can have infinite entropy. However such distributions are highly pathological, and unlikely to occur in a natural context \cite{baccetti_infinite_2013}. We will therefore add the requirement that the parameter $\phi$ must have finite entropy.

This requirement is sufficient to ensure that all other quantities are finite. For the mutual information we have \cite{wilde_quantum_2017}
\begin{equation}
  I(\Phi;M)=H(\Phi)-H(\Phi|M)\le H(\Phi),
\end{equation}
while an ensemble $\mathcal{E}_{\Phi}$ encoding a parameter $p_{\Phi}(\phi)$ cannot contain more information than $p_{\Phi}(\phi)$ itself \cite[Eq. (1)]{shirokov_upper_2019}:
\begin{equation}
  \chi(\mathcal{E}_{\Phi})\le H(\Phi).
\end{equation}
Since all other quantities defined in the proof may be expressed in terms of the above, all sums are guaranteed to be finite. Thus the CXI equality also holds for continuous parameters and in infinite dimensional Hilbert spaces.

From the CXI equality, we can see that the optimal projective measurement is the one which minimises the ensemble coherence. Most of the time there does not exist a basis such that $C_M(\mathcal{E}_{\Phi})=0$, meaning that some information will always be `locked away' in the coherences. Accessing this information will require a collective measurement on multiple probes. Precisely how much advantage such a scheme would provide is quantified exactly by the minimum value of the ensemble coherence.

\section{General POVM measurements}\label{sec:CXIPOVM}
Projective measurements do not describe all measurement schemes. A common example is photodetection, which does not project the system into an eigenstate. More generally, measurements can be modelled as entangling our state with an ancilla, and then performing a projective measurement on the combined system. Such measurements are described by the framework of Positive Operator Valued Measures (POVMs) \cite{wiseman_quantum_2009}. In this section we will generalise \cref{thm:CXI} to POVMs, using the recent extension of the relative entropy of coherence \cite{bischof_resource_2019,bischof_quantifying_2021}.

A POVM is described by a set of positive semi-definite operators $\{M_j\}$ such that 
\begin{equation}\label{eq:POVMIdentitySum}
    \sum_jM_j^{\dagger}M_j=I,
\end{equation}
where $I$ is the identity operator. The probability of the $j$th measurement outcome is given by
\begin{equation}\label{eq:POVMProbability}
  p_j=\mathrm{tr}\{\rho M_j^{\dagger}M_j\}.
\end{equation}
Positivity of the measurement operators ensure that the probabilities are positive, and from \cref{eq:POVMIdentitySum} the probabilities sum to unity. If the $j$th measurement outcome is detected, the post-measurement state is
\begin{equation}\label{eq:POVMPostState}
  \rho_j=M_j\rho M_j^{\dagger}/p_j.
\end{equation}
If $\{M_j\}$ is a set of orthogonal projectors, then $M_j^{\dagger}M_j=M_j$ and we recover the usual expressions for a projective measurement. 

We will now describe the Naimark dilation, which lets us represent every POVM as a projective measurement on a larger Hilbert space \cite{watrous_theory_2018}. Let $\mathcal{H}$ be the Hilbert space of our system, and $\mathcal{A}$ be an ancilla space of dimension equal to the number $N$ of measurement operators in the POVM. We define the map $V:\mathcal{H}\rightarrow\mathcal{H}\otimes\mathcal{A}$ as:
\begin{equation}
  V = \sum_jM_j\otimes |j\rangle.
\end{equation}
This lifts our state $\rho\in\mathcal{H}$ to $V\rho V^{\dagger}$ in the larger Hilbert space $\mathcal{H}\otimes\mathcal{A}$. Going forward, we will use tildes to refer to quantities in our expanded Hilbert space:
\begin{equation}
  \tilde{\rho}=V\rho V^{\dagger}.
\end{equation}

\cref{eq:POVMIdentitySum} implies $V$ is an isometry, meaning for all states $\lvert\psi\rangle,\lvert\phi\rangle\in\mathcal{H}$ we have $\langle\psi|\phi\rangle=\langle\tilde{\psi}|\tilde{\phi}\rangle=\left(\langle\psi| V^{\dagger}\right)\left(V|\phi\rangle\right)$. To see this we expand:
\begin{equation}
  \begin{aligned}
    \langle\psi| V^{\dagger}V|\phi\rangle &= \langle\psi|\sum_{jk}M_j^{\dagger}M_k\langle k|j\rangle|\phi\rangle, \\
      &=\langle\psi|\sum_jM_j^{\dagger}M_j|\phi\rangle, \\
      &=\langle\psi|\phi\rangle,
  \end{aligned}
\end{equation}
where in the second-last line we used \cref{eq:POVMIdentitySum}. 
Thus $\tilde{\rho}$ has the same eigenvalues as $\rho$. Since entropy is a function of the eigenvalues of a state, the dilation $\tilde{\rho}$ has the same entropy as the original state $\rho$.

In the dilated Hilbert space the projection operator corresponding to $M_j$ is
\begin{equation}
  \tilde{M}_j=I_{\mathcal{H}}\otimes |j\rangle\langle j|_{\mathcal{A}},
\end{equation}
where $I_{\mathcal{H}}$ is the identity operator on $\mathcal{H}$ and $\lvert j\rangle_{\mathcal{A}}$ is the $j$th basis element of $\mathcal{A}$. We will omit the subscript Hilbert spaces from now on. The probability of the $j$th measurement is given by \cref{eq:POVMProbability}:
\begin{equation}
  \begin{aligned}
    \tilde{p}_j &= \mathrm{tr}\{\tilde{\rho}\tilde{M}_j\}, \\
      &= \mathrm{tr}\left\{(V\rho V^{\dagger})(I\otimes |j\rangle\langle j|)\right\}, \\
      &= \sum_{kl}\mathrm{tr}\left\{M_k\rho M_l^{\dagger}\otimes \left(|k\rangle\langle l|j\rangle\langle j|\right)\right\}, \\
      &=\mathrm{tr}\{\rho M_j^{\dagger}M_j\}, \\
      &= p_j.
  \end{aligned}
\end{equation}
Thus the dilated measurement has the same statistics as the POVM.

Naimark's dilation can be used to generalise the relative entropy of coherence to POVMs \cite{bischof_resource_2019}. Given a state $\rho$, the coherence relative to a POVM $M=\{M_j\}$ is defined as the coherence of $\tilde{\rho}$ relative to the projective operators $\{\tilde{M}_j\}$:
\begin{equation}\label{eq:POVMCoherence}
  \begin{aligned}
    C_M(\rho) &= S\left(\Delta_{\tilde{M}}[\tilde{\rho}]\right)-S(\tilde{\rho}), \\
      &= S\left(\sum_j\tilde{M}_j\tilde{\rho}\tilde{M}_j^{\dagger}\right)-S(\tilde{\rho}).
  \end{aligned}
\end{equation}
It is shown in \cite{bischof_resource_2019} this is well-defined and satisfies the requirements of a coherence measure. When $M$ is a projective measurement this reduces to the relative entropy of coherence, thus we are justified in using the same symbol $C_M$ for both. Going forward, we will refer to this as the \emph{POVM coherence}.

We are now prepared to generalise the CXI equality to positive-operator valued measures:
\begin{thm}[CXI equality for POVMs]\label{thm:CXIPOVM}
  Let $\Phi$ be a parameter with probability distribution $p_{\Phi}(\phi)$, with a corresponding ensemble $\mathcal{E}_{\Phi}$ of states $\rho_{\phi}$. If we perform a single POVM measurement on the ensemble, the CXI equality holds:
  \begin{equation}\label{eq:CXIStatementPOVM}
    C_M(\mathcal{E}_{\Phi})=\chi(\mathcal{E}_{\Phi})-I(\Phi;M),
  \end{equation}
  where $C_M(\mathcal{E}_{\Phi})$ is the ensemble coherence as defined in \cref{eq:EnsembleCoherence}, and the coherence measure is the POVM coherence defined in \cref{eq:POVMCoherence}.
\end{thm}

\begin{proof}[Proof of CXI for POVMs]
  Let $\tilde{\mathcal{E}}_{\Phi}$ be the dilated ensemble consisting of states $\tilde{\rho}_{\phi}$ with probability $p_{\Phi}(\phi)$. Since $\tilde{M}=\{\tilde{M}_j\}$ is a projective measurement of $\tilde{\mathcal{E}}_{\Phi}$, the CXI equality holds for the dilated measurement:
\begin{equation}
  C_{\tilde{M}}(\tilde{\mathcal{E}}_{\Phi})=\chi(\tilde{\mathcal{E}}_{\Phi})-I(\tilde{M},\Phi).
\end{equation}
  The term on the left-hand side is the POVM coherence $C_M(\mathcal{E}_{\Phi})$. Let us thus consider the right-hand side.

  The Holevo information is given by
  \begin{equation}
    \chi(\tilde{\mathcal{E}}_{\Phi})=S\left(\int\tilde{\rho}_{\phi}p_{\Phi}(\phi)\dd\phi\right)-\int S(\tilde{\rho}_{\phi})p_{\Phi}(\phi)\dd\phi.
  \end{equation}
  Since the dilation preserves entropy, for the second term we have $S(\tilde{\rho}_{\phi})=S(\rho_{\phi})$. Thus let us examine the integral in brackets:
  \begin{equation}
    \begin{aligned}
      \int\tilde{\rho}_{\phi}p_{\Phi}(\phi)\dd\phi &= \int\left(V\rho_{\phi}V^{\dagger}\right)p_{\Phi}(\phi)\dd\phi, \\
        &= V\left(\int \rho_{\phi}p_{\Phi}(\phi)\dd\phi\right)V^{\dagger}.
    \end{aligned}
  \end{equation}
  Since the dilation $V$ doesn't change the entropy, we find that the Holevo information of the dilated ensemble is the same as the original ensemble:
  \begin{equation}
    \chi(\tilde{\mathcal{E}}_{\Phi})=\chi(\mathcal{E}_{\Phi}).
  \end{equation}
  
  Now let us consider the mutual information:
  \begin{equation}
    I(\tilde{M},\Phi)=H(\tilde{M})-H(\tilde{M}|\Phi),
  \end{equation}
  where $H$ is the Shannon entropy. The first term $H(\tilde{M})$ is the entropy of the probability distribution of measurement results. Since the dilation has the same measurement statistics, we immediately have $H(\tilde{M})=H(M)$. 
  
  Now let's consider the second term:
  \begin{equation}
    H(\tilde{M}|\Phi)=\sum_{\phi}p_{\Phi}(\phi)H(\tilde{M}|\phi).
  \end{equation}
  To calculate this, we need to find the measurement probabilities of the dilated conditional state $\tilde{\rho}_{\phi}$. These are:
  \begin{equation}
    \begin{aligned}
      \tilde{p}_{j|\phi} &= \mathrm{tr}\left\{\tilde{M}_j\tilde{\rho}_{\phi}\right\}, \\
        &= \mathrm{tr}\left\{\left(I\otimes |j\rangle\langle j|\right)\left(V\rho_{\phi}V^{\dagger}\right)\right\}, \\
        &= \mathrm{tr}\left\{\left(I\otimes |j\rangle\langle j|\right)\left(\sum_{kl}M_k\rho_{\phi}M_l^{\dagger}\otimes |k\rangle\langle l|\right)\right\}, \\
        &= \mathrm{tr}\left\{M_j\rho_{\phi} M_j^{\dagger}\right\}, \\
        &= p_{j|\phi}.
    \end{aligned}
  \end{equation}
  Since the probability distributions are the same, we thus have $H(\tilde{M}|\Phi)=H(M|\Phi)$. Thus we find that the mutual information between our measurement and the parameter is the same for both the original and dilated ensemble:
  \begin{equation}
    I(\tilde{M},\Phi)=I(M,\Phi).
  \end{equation}
  Putting all of the above together gives the CXI equality \cref{eq:CXIStatementPOVM}.
\end{proof}

Thus using the POVM coherence, we find that the CXI equality holds for general quantum measurements.

\section{Example}\label{sec:Measurement}
Let us now consider a brief example to illustrate the ensemble coherence and CXI equality. Suppose we wish to discriminate between two pure states on the Bloch sphere. We take the first state to lie along the $\sigma_x$-axis: 
\begin{equation}
  \rho_0=\lvert\uparrow_x\rangle.
\end{equation}
The second state $\rho_\theta$ is rotated an angle $\theta$ from $\rho_0$ about the $\sigma_z$ axis:
\begin{equation}\label{eq:RotatedState}
  \rho_{\theta}=\cos(\theta/2)\lvert\uparrow_x\rangle+\sin(\theta/2)\lvert\downarrow_x\rangle.
\end{equation}
In this problem of state discrimination, the parameter $\phi$ we are estimating can take two values, $0$ and $\theta$, corresponding to the two possible states $\rho_0,\rho_{\theta}$. The two states are assumed to have equal probability. We will consider three different scenarios, with $\theta$ equal to $\pi/10$, $\pi/2$ and $\pi$.

Let us first study discrimination via projective measurement. A projective measurement on the Bloch sphere is described by a pair of antipodal points representing the basis states of the measurement. From the CXI equality, the ensemble coherence with respect to this basis tells us how efficient measurement in this basis is at discriminating between the two states. The greater the coherence, the more information is being `lost' by the measurement. We note however that the three scenarios will have different Holevo informations. For the perfectly distinguishable case of $\theta=\pi$, the Holevo information is $\chi=\log 2\approx 0.69$. For $\theta=\pi/2$ we have $\chi\approx 0.42$, and for the small separation $\theta=\pi/10$ the Holevo information is $\chi\approx 0.04$.\footnote{If we evaluate the logarithm in \cref{eq:HolevoInformation} in base 2, the Holevo informations are $\chi=1\,\mathrm{bit}$ for $\theta=\pi$, $\chi\approx0.6\,\mathrm{bits}$ for $\theta=\pi/2$, and $\chi\approx 0.05\,\mathrm{bits}$ for $\theta=\pi/10$. The other numbers reported in this section are all ratios of entropies, which are independent of the logarithm base.} When evaluating the efficacy of a measurement basis, we must compare the ensemble coherence with the total Holevo information available.

\begin{figure}[h]
  \centering
  \includegraphics[width=\columnwidth]{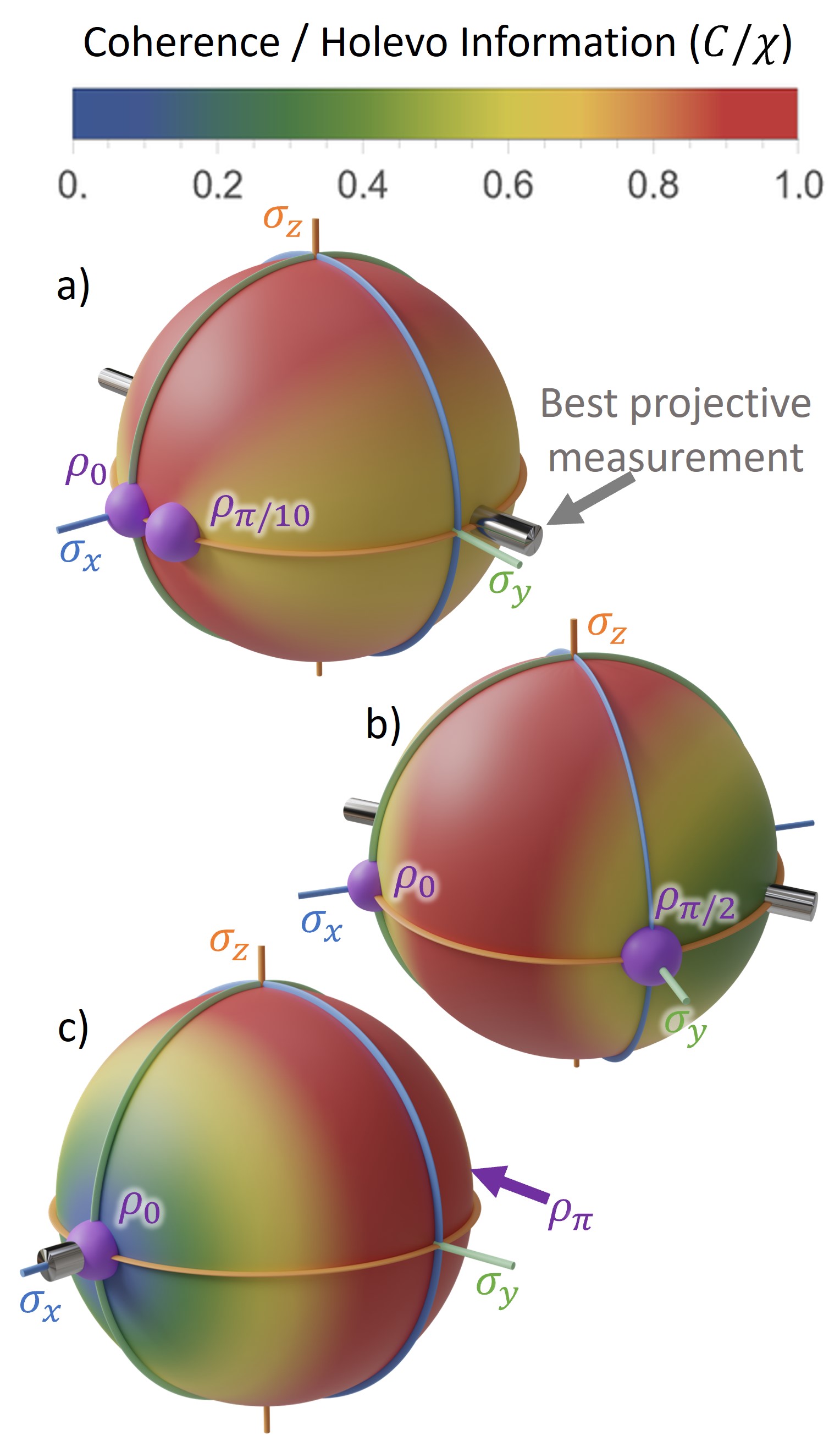}
  \caption{Ensemble coherence in state discrimination. The states $\rho_0,\rho_{\theta}$ are separated by an angle a) $\pi/10$, b) $\pi/2$, and c) $\pi$. Points on the Bloch sphere are coloured by the ensemble coherence of a projective measurement in the corresponding basis, normalised by the Holevo information. From the CXI equality, a larger coherence indicates more lost information, which we can see correlates with bases that do not distinguish between the two states. The silver bar denotes the basis with minimum ensemble coherence, which coincides with the basis given by minimum-error state discrimination theory. As the states grow more distinguishable the coherence in the optimum basis increases, quantifying the advantage that could be gained from a multi-partite measurement.}\label{fig:BlochSphere}
\end{figure}

In \cref{fig:BlochSphere} we graph the ensemble coherence in each possible measurement basis on the Bloch sphere, normalised by the Holevo information. We can see that measurements whose outcomes do not discriminate between the two states have coherence equal to the Holevo information, indicating that they provide no information. Meanwhile the bases with minimum coherence make an angle $\theta/2+\pi/2$ with the $\sigma_x$-axis. This corresponds to the Helstrom measurement from the theory of minimum-error quantum state discrimination \cite{bergou_discrimination_2010,helstrom_quantum_1969}. 

\begin{figure}
  \centering
  \includegraphics[width=\columnwidth]{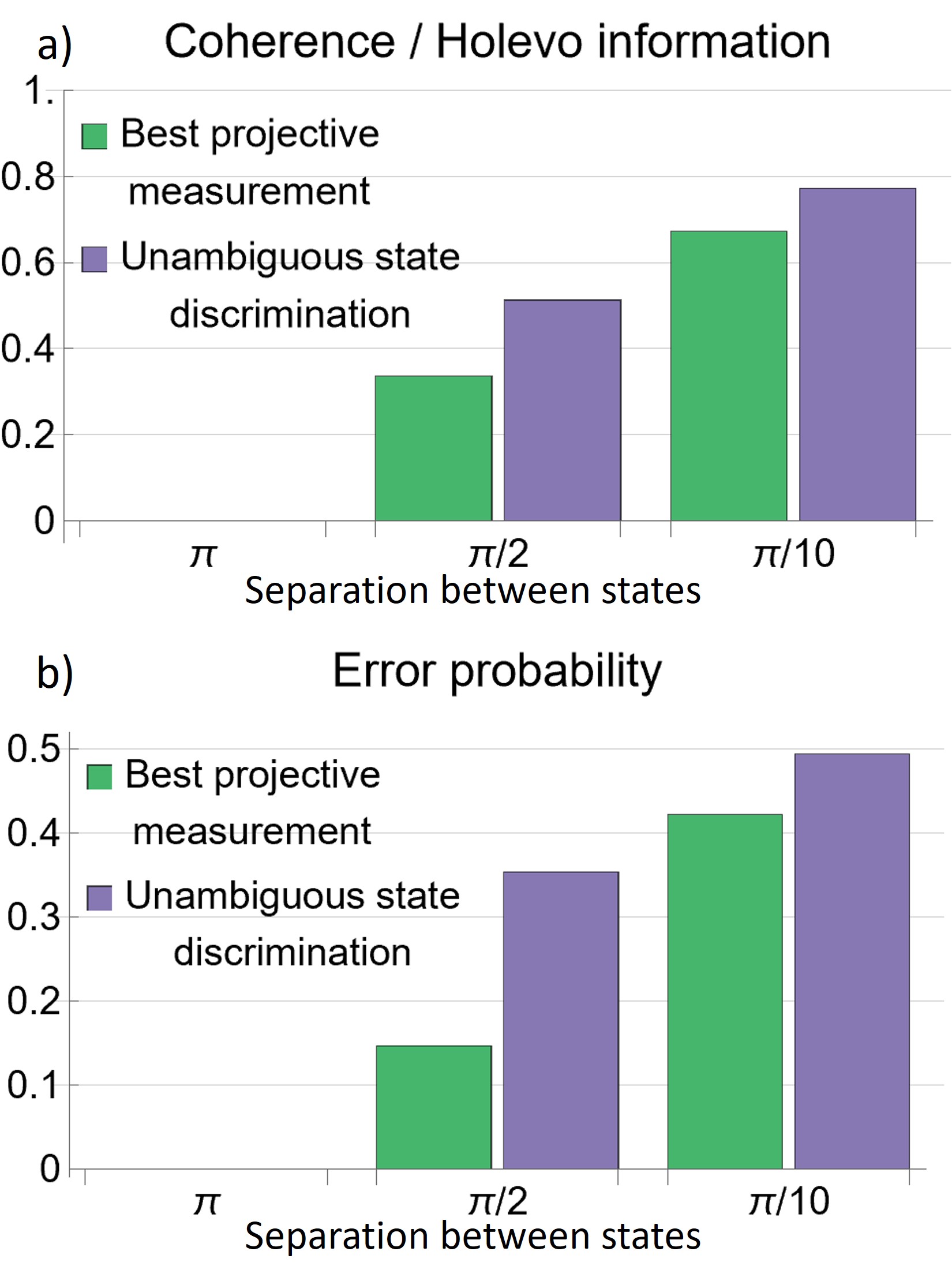}
  \caption{a) We compare the ensemble coherence of the best projective measurement, with the ensemble coherence of the POVM for unambiguous state discrimination. Coherences are normalised with respect to the Holevo information. When the states are separated by an angle of $\pi$ both coherences are zero, since the measurements can perfectly discriminate between the states. As the states grow less distinguishable however these measurements lose information, as quantified by the increasing coherece. In general the POVM has larger coherence, indicating it will take more measurements on average to successfully discriminate between the states. b) The error probability if we attempt to discriminate between the two states based on a single measurement result. We can see that an increase in the ensemble coherence correlates with an increase in error probability.}\label{fig:CoherencePlots}
\end{figure}

Now let us study the ensemble coherence for a POVM. In unambiguous state discrimination we consider a three-element measurement, with outcomes corresponding to $\rho_0$, $\rho_{\theta}$, and `unsure' \cite{barnett_quantum_2009,bergou_discrimination_2010}. The first two outcomes correspond respectively to measurement operators:
\begin{equation}
  \begin{aligned}
    M_0 &= c(I-\rho_{\theta})=c\rho_{\theta+\pi}, \\
    M_{\theta} &= c(I-\rho_{0}) = c\rho_{\pi}.
  \end{aligned}
\end{equation}
Here $c$ is a constant to be determined, which is the same for both $M_0,M_{\theta}$ since the problem is symmetric with regards to the two states. The measurement operator corresponding to the unsure outcome is then
\begin{equation}
  \begin{aligned}
    M_{?}  &= I-M_0-M_{\theta}, \\
      &= I-c\left(\rho_{\theta+\pi}+\rho_{\pi}\right).\label{eq:PiZero}
  \end{aligned}
\end{equation}
The probability of this outcome is
\begin{align}
  p_{?} &= \tr{M_{?}\frac{1}{2}\left(\rho_0+\rho_{\theta}\right)}, \\
    &= \frac{1}{2}\tr{\left(\rho_0+\rho_{\theta}\right)-c\left(\rho_{\theta+\pi}\rho_0+\rho_{\pi}\rho_{\theta}\right)}, \\
    &= 1-c\sin^2\left(\frac{\theta}{2}\right).\label{eq:PQuestion}
\end{align}
To minimise \cref{eq:PQuestion}, we must make $c$ as large as possible while preserving positivity of the operator $M_{?}$ in \cref{eq:PiZero}. This is satisfied by $c=1/\lambda$, where $\lambda$ is the largest eigenvalue of $\rho_{\pi}+\rho_{\theta+\pi}$.

In \cref{fig:CoherencePlots} a) we show the ensemble coherences for the POVM with measurement operators $\{M_0,M_{\theta},M_{?}\}$, and compare these with the optimal projective measurement. We can see that apart from the perfectly distinguishable case of $\theta=\pi$ when both are zero, the coherence of unambiguous state discrimination is larger. Thus we will on average require more measurements to distinguish the states using unambiguous state discrimination, than by the optimal projective measurement.

Let us consider the magnitudes of the coherence values in \cref{fig:CoherencePlots} a). For $\theta=\pi$, where $\rho_0,\rho_{\theta}$ are orthogonal, the coherence is zero for both projective measurements and the POVM we considered. Thus no better measurement is possible than projective measurement on a single system. As the states grow increasingly indistinguishable however, the coherence increases to a sizeable fraction of the Holevo information. This means that a protocol which performed a collective measurement on multiple probes could obtain a substantial information gain per probe. Roughly speaking, this is possible because the states $\rho_0^{\otimes n},\rho_{\theta}^{\otimes n}$ are `more orthogonal' than $\rho_0,\rho_{\theta}$, making it easier to discriminate between them.\footnote{The fidelity between $\rho_0^{\otimes n}$ and $\rho_{\theta}^{\otimes n}$ decreases exponentially with the number $n$ of probes, as can be seen by computing $(\langle\psi\rvert)^{\otimes n}(\lvert\phi\rangle)^{\otimes n}=(\langle\psi|\phi\rangle)^n$.}

As a demonstration that ensemble coherence does indeed correspond to lost information, we will explicitly compute the success probabilities when performing state discrimination based off a single measurement. In the projective case the two measurement operators for the basis with minimum coherence are
\begin{equation}
  \begin{aligned}
    \Pi_0 &= \rho_{\theta/2-\pi/2}, \\
    \Pi_{\theta} &= \rho_{\theta/2+\pi/2}.
  \end{aligned}
\end{equation}
We projectively measure our state in this basis, and estimate the parameter to be $0$ if we observe $\Pi_0$, or $\theta$ if we observe $\Pi_{\theta}$. The success probability $p_{\mathrm{s}}$ is then:
\begin{equation}
  \begin{aligned}
    p_{\mathrm{s}}&= p_{\Phi}(0)p_{M|\Phi}(0|0)+p_{\Phi}(\theta)p_{M|\Phi}(\theta|\theta), \\
      &= \frac{1}{2}\left(\tr{\Pi_0\rho_0}+\tr{\Pi_{\theta}\rho_{\theta}}\right).
  \end{aligned}
\end{equation}
For the case of unambiguous state discrimination, as before we estimate a parameter value of $0$ or $\theta$ if we observe $M_0$ or $M_{\theta}$ respectively. If we observe the third outcome $M_{?}$ we gain no information, and our best strategy is to randomly guess either $0$ or $\theta$, in which case we will be correct half the time. Our success probability is then
\begin{equation}
  \begin{aligned}
    p_{\mathrm{s}}&= p_{\Phi}(0)\left(p_{M|\Phi}(0|0)+\frac{1}{2}p_{M|\Phi}(?|0)\right) \\
    &\phantom{=}+p_{\Phi}(\theta)\left(p_{M|\Phi}(\theta|\theta)+\frac{1}{2}p_{M|\Phi}(?|\theta)\right), \\
      &= \frac{1}{2}\left(\tr{\left(M_0+\frac{1}{2}M_{?}\right)\rho_0}+\tr{\left(M_{\theta}+\frac{1}{2}M_{?}\right)\rho_{\theta}}\right).
  \end{aligned}
\end{equation}
Since coherence represents lost information, it is appropriate to compare this with the error probability $1-p_{\mathrm{s}}$. We graph this in \cref{fig:CoherencePlots} b) for both projective measurement and the POVM. We can see that unambiguous state discrimination is more likely to fail than projective measurement, as expected by its larger ensemble coherence. The probability of error also increases as the separation between the states decreases, again in line with the increasing coherence.

This examples demonstrates that the ensemble coherence does indeed represent information `locked away in the coherences', and inaccessible to the chosen measurement. Graphing this allows us to visualise and compare different measurement schemes. In the supplementary material we study a more complex example of adaptive measurement. All code and data for the results in this section and the supplementary material can be downloaded from \cite{lecamwasam_cxi_2023}

\nocite{demkowicz-dobrzanski_quantum_2015,tsang_quantum_2016,weisstein_hermite_nodate,trees_detection_2013}
\section{Conclusion}\label{sec:Conclusion}
This work establishes a fundamental connection between coherence and Bayesian metrology. 
The information-theoretic tools provide a general proof, which holds regardless of how the parameter is encoded in the state, and for both projective and POVM measurements. In particular the CXI equality applies even to discontinuous situations such as state discrimination, where the quantum Fisher information is not applicable.

An obvious direction for generalisation of this work would be to study multiparameter estimation \cite{demkowicz-dobrzanski_multi-parameter_2020}. It would also be fruitful to see what light existing results from the resource theory of coherence could shed on metrology. For example, given the current interest in metrology for fundamental tests of physics \cite{bose_massive_2023,moore_searching_2021}, coherence-based measures of macroscopicity could help understand how quantum effects could be observed on large scales \cite{yadin_general_2016,kwon_coherence_2018,hertz_quadrature_2020}.

There is an intuitive reason for why the coherence measure in the CXI relation is the relative entropy of coherence. Classically, the relative entropy between two probability distributions $p$ and $q$ can be thought of as the information lost if $q$ is used to approximate $p$ \cite[\S 2.1]{anderson_model_2002}. A projective measurement in the basis $M$ effectively approximates $\rho$ with the decohered state $\Delta_M[\rho]$, since all information contained in the off-diagonal elements is lost. Moreover, the appearance of the POVM coherence \cite{bischof_resource_2019} when we generalised the CXI equality shows that it is indeed a natural generalisation of the standard relative entropy, and provides a novel operational interpretation. However, there are many other measures of coherence \cite{streltsov_colloquium_2017}. It would be interesting to study their meanings when applied to ensembles along the lines of \cref{eq:EnsembleCoherence}.



\section{Acknowledgements}
We are especially grateful to Michael Hall for introducing us to Bayesian quantum metrology, as well as for fruitful discussions with Simon Haine, Stuart Szigeti, and Zain Mehdi, and comments on the manuscript by Lorc\'an Conlon, Simon Yung, and the anonymous referees. This research was funded by the Australian Research Council Centre of Excellence for Quantum Computation and Communication Technology (Grant No. CE110001027), the Singapore Ministry of Education Tier 2 (Grant No. MOE-T2EP50221- 0005),  the Singapore Ministry of Education Tier 1 (Grant No. RG77/22), the National Research Foundation Singapore, and the Agency for Science, Technology and Research (A*STAR) under its QEP2.0 programme (Grant No. NRF2021-QEP2-02-P06), and Grant No. FQXI-RFP-IPW-1903 `Are quantum agents more energetically efficient at making predictions?' from the Foundational Questions Institute and Fetzer Franklin Fund (a donor-advised fund of Silicon Valley Community Foundation). RL was supported by an Australian Government Research Training Program (RTP) scholarship. Calculations for an early version of this work were performed on the Deigo cluster computer at the Okinawa Institute of Science and Technology.  


%

\clearpage

\end{document}


\author{Ruvi Lecamwasam\,\orcidlink{0000-0001-6531-3233}}
\email{me@ruvi.blog}
\affiliation{A*STAR Quantum Innovation Centre (Q.Inc), Institute of Materials Research and Engineering (IMRE), Agency for Science, Technology and Research (A*STAR), 2 Fusionopolis Way, 08-03 Innovis 138634, Singapore \looseness=-1}
\affiliation{Quantum Machines Unit, Okinawa Institute of Science and Technology Graduate University, Onna, Okinawa 904-0495, Japan \looseness=-1}
\affiliation{Centre for Quantum Computation and Communication Technology, Department of Quantum Science and Technology, Australian National University, ACT 2601, Australia \looseness=-1}
\author{Syed Assad\,\orcidlink{0000-0002-5416-7098}}
\affiliation{A*STAR Quantum Innovation Centre (Q.Inc), Institute of Materials Research and Engineering (IMRE), Agency for Science, Technology and Research (A*STAR), 2 Fusionopolis Way, 08-03 Innovis 138634, Singapore \looseness=-1}
\affiliation{Centre for Quantum Computation and Communication Technology, Department of Quantum Science and Technology, Australian National University, ACT 2601, Australia \looseness=-1}
\author{Joseph J. Hope\,\orcidlink{0000-0002-5260-1380}}
\affiliation{Department of Quantum Science and Technology, Australian National University, ACT 2601, Australia \looseness=-1}
\author{\\Ping Koy Lam\,\orcidlink{0000-0002-4421-601X}}
\affiliation{A*STAR Quantum Innovation Centre (Q.Inc), Institute of Materials Research and Engineering (IMRE), Agency for Science, Technology and Research (A*STAR), 2 Fusionopolis Way, 08-03 Innovis 138634, Singapore \looseness=-1}
\affiliation{Centre for Quantum Technologies, National University of Singapore, 3 Science Drive 2, Singapore 117543 \looseness=-2}
\affiliation{Centre for Quantum Computation and Communication Technology, Department of Quantum Science and Technology, Australian National University, ACT 2601, Australia \looseness=-1}
\author{Jayne Thompson\,\orcidlink{0000-0002-3746-244X}}
\email{thompson.jayne2@gmail.com}
\affiliation{Horizon Quantum Computing, 05-22 Alice@Mediapolis, 29 Media Circle, Singapore 138565 \looseness=-1}
\affiliation{Institute of High Performance Computing, Agency for Science, Technology, and Research (A*STAR), Singapore 138634, Singapore \looseness=-2}
\author{Mile Gu\,\orcidlink{0000-0002-5459-4313}}
\email{mgu@quantumcomplexity.org}
\affiliation{Nanyang Quantum Hub, School of Physical and Mathematical Sciences, Nanyang Technological University,  21 Nanyang Link, Singapore 639673 \looseness=-2}
\affiliation{Centre for Quantum Technologies, National University of Singapore, 3 Science Drive 2, Singapore 117543 \looseness=-2}
\affiliation{MajuLab, CNRS-UNS-NUS-NTU International Joint Research Unit, UMI 3654, 117543, Singapore\looseness=-2}

\date{\today}

\title{Supplementary to `The relative entropy of coherence quantifies performance in Bayesian metrology'}

\begin{abstract}
In this supplementary material we provide another example of the ensemble coherence in metrology, applied to a problem of adaptive measurement in the Hermite-Gauss basis. We show how one can design a measurement scheme by repeatedly minimising the ensemble coherence, which is equivalent to maximising the mutual information. Code generating all results from this section may be downloaded from \cite{lecamwasam_cxi_2023}.
\end{abstract}

\maketitle

\setcounter{section}{0}
\setcounter{equation}{0}
\setcounter{figure}{0}
\setcounter{table}{0}
\setcounter{page}{1}
\renewcommand{\theequation}{S\arabic{equation}}
\renewcommand{\thefigure}{S\arabic{figure}}
\renewcommand{\thesection}{S\arabic{section}}

\section{Adaptive measurement example}\label{sec:Measurement}
To illustrate the ensemble coherence and CXI equality, we will apply them to an example of parameter estimation. Graphing the ensemble coherence will allow us to compare the information gain from different measurement bases. Repeatedly choosing the basis with minimum coherence will provide an optimal adaptive measurement schemes. Moreover, the absolute value of the coherence quantifies how much more information could be gained using more resources, such as other bases or multipartite measurements. As previously we will take the parameter to be discrete, but the continuous case is analogous with sums replaced by integrals. Detailed calculations will be left to \cref{app:AdaptiveMeasurement}.

Many common metrological scenarios involve estimation of a phase. Analysis of these situations however is complicated by the topology of the circle, which requires that phases of $\epsilon$ and $2\pi-\epsilon$ be `close' for small $\epsilon$. For simplicity we consider here estimation of a linear parameter. The ensemble coherence and CXI equality can still be applied to phase estimation, only more care is required to properly quantify the error in the estimate \cite[\S IV.A.2]{demkowicz-dobrzanski_quantum_2015}. 

We will consider estimating the location of a point source of light, via photodetection in the basis of Hermite-Gauss modes. The wavefunction of a photon emitted from a Gaussian source centered at $\phi$ with unit variance is
\begin{equation}\label{eq:PointSourceWF}
  \psi_{\phi}(x)=\frac{1}{(2\pi)^{1/4}}\exp\left(-\frac{(x-\phi)^2}{4}\right).
\end{equation}
Here $x$ is the position of the emitted photon, and $\phi$ is the location of the source, which is the parameter we will attempt to estimate. Then a photon emitted from a point source at $\phi$ has state $\rho_{\phi}=\ketbra{\psi_{\phi}}{\psi_{\phi}}$, where $\ket{\psi_{\phi}}=\int\psi_{\phi}(x)\ket{x}\mathrm{d}x$. 

Our prior knowledge of the location of the source is represented by a probability distribution $p_{\Phi}(\phi)$. We take this to be a mixture of discrete Gaussians centered on $\phi=\pm 1$ with variance $0.5$, i.e. half the width of the photon wavefunction. This prior is plotted in \cref{fig:Adaptive}a. Such a double-peaked structure will benefit from adaptive strategies, as well as entangled probes which can make use of interference. This thus provides a good demonstration of what can be learned from the ensemble coherence. 

The ensemble state $\rho_{\Phi}$ of an emitted photon is then
\begin{equation}
  \rho_{\Phi}=\sum_{\phi}p_{\Phi}(\phi)\rho_{\phi}.
\end{equation}
When we gain information from a measurement, we update the probability distribution of $\phi$ using Bayes' rule, and hence update the ensemble corresponding to the next photon. The information that we gain will depend on the choice of measurement basis.  A natural choice may be photodetection in the position basis. However, we will consider a less well-matched basis, and use the ensemble coherence to construct an optimum sequence of adaptive measurements.

Recent work in superresolution imaging has demonstrated that the separation between two closely-spaced point sources may be optimally estimated by sampling in the basis of \emph{Hermite-Gauss (HG) modes} \cite{tsang_quantum_2016}. For integer $q$ these have wavefunction
\begin{equation}\label{eq:HermiteGaussModes}
  h_q(x)=\frac{1}{(2\pi\sigma_h^2)^{1/4}}\frac{1}{\sqrt{2^qq!}}H_q\left(\frac{x}{\sqrt{2}\sigma_h}\right)e^{-x^2/4\sigma_h^2},
\end{equation}
where $\sigma_h$ is the width of the mode, and $H_q$ is the $q$th Hermite polynomial \cite{weisstein_hermite_nodate}. The HG modes form an orthogonal basis of wavefunctions. The first part of the superresolution procedure involves estimating the centre of the point sources. Thus suppose that we had constructed a superresolution apparatus, and wished to use it to estimate the centre of a single point source. This corresponds to projective measurement in the basis of Hermite-Gauss modes. To reflect imperfect experimental apparatus we take $\sigma_h=2$. Thus our sampling basis has twice the width of the photon wavefunctions, making it harder to infer the source location.

\begin{figure}[!]
  \centering
    \begin{tikzpicture}
        \node[anchor=south west,inner sep=0] (imgHGModes) {\includegraphics[width=\columnwidth]{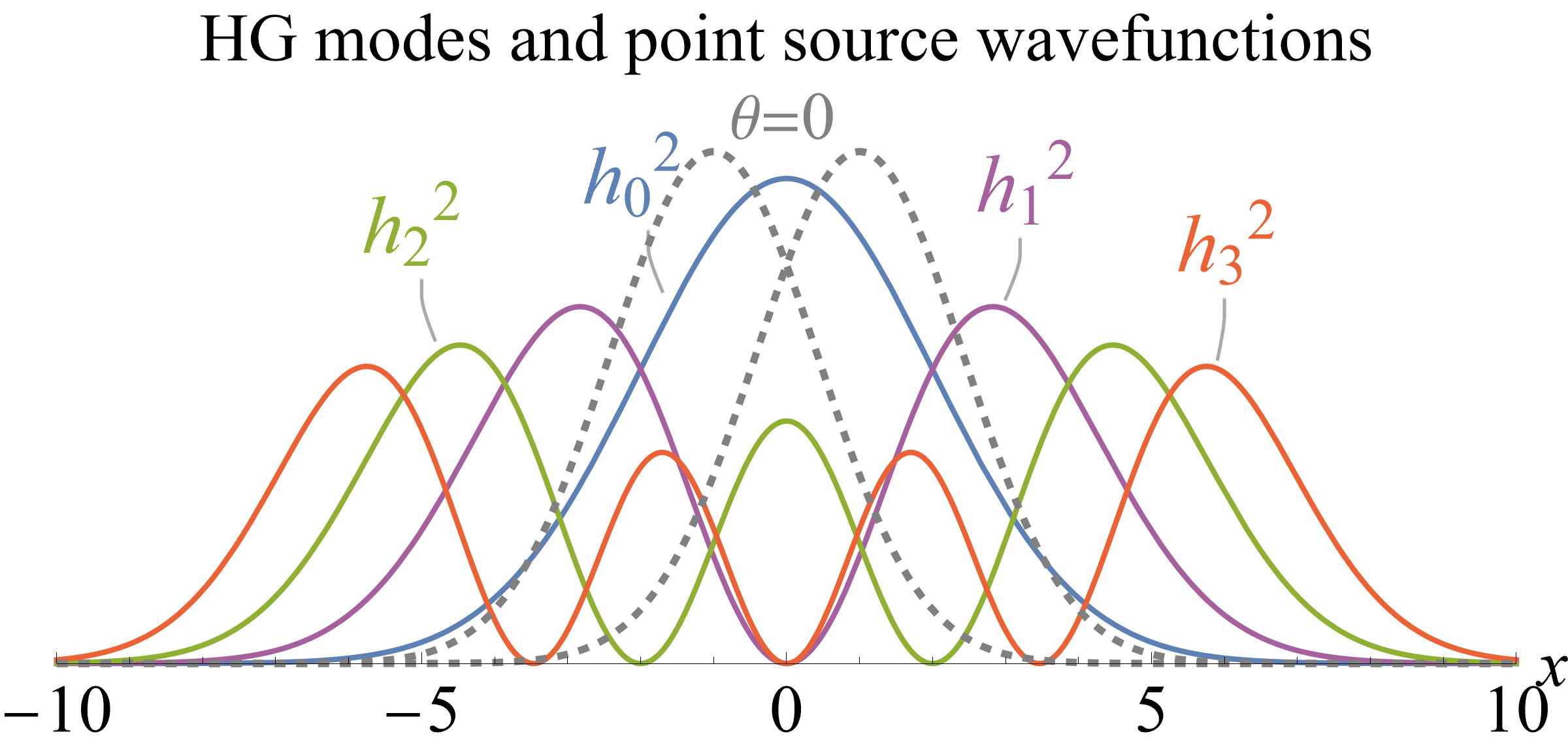}};
        \begin{scope}[x={(imgHGModes.south east)},y={(imgHGModes.north west)}]
            \node[anchor=north west] at (0,1) {a)};
        \end{scope}
    \end{tikzpicture}%
    \vspace{1em}
    \begin{tikzpicture}
        \node[anchor=south west,inner sep=0] (imgHGOptimal) {\includegraphics[width=\columnwidth]{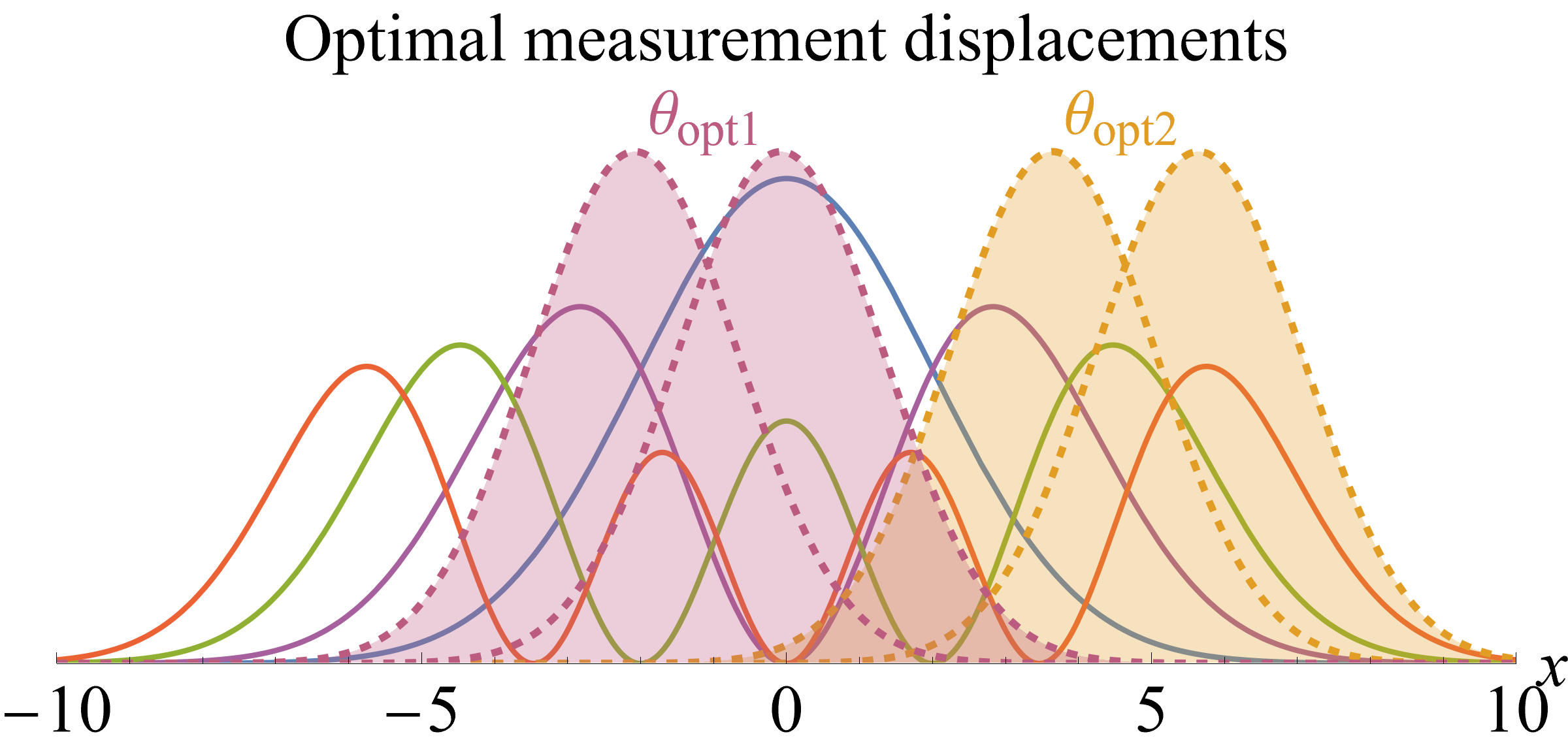}};
        \begin{scope}[x={(imgHGOptimal.south east)},y={(imgHGOptimal.north west)}]
            \node[anchor=north west] at (0,1) {b)};
        \end{scope}
    \end{tikzpicture}
  \caption{a) Solid lines denote the squared Hermite-Gauss modes defined in \cref{eq:HermiteGaussModes}, which form the measurement basis. As $q$ grows larger there are more oscillations and the mode extends further along the $x$-axis. The dashed lines show photon wavefunctions \cref{eq:PointSourceWF} emitted from sources at $\phi=-1$ and $\phi=+1$. Due to symmetry of the HG modes about zero, measurements cannot discriminate between these. b) We can break the symmetry by shifting our measurement basis by $\theta$. Shifts $\theta_{\mathrm{opt}1},\theta_{\mathrm{opt}2}$ minimise the ensemble coherence of the prior distribution. These position the wavefunctions so that the HG modes optimally differentiate between source locations.}\label{fig:HGModes}
\end{figure}

\begin{figure}[!]
  \centering
    \begin{tikzpicture}
        \node[anchor=south west,inner sep=0] (imgPDists) {\includegraphics[width=\columnwidth]{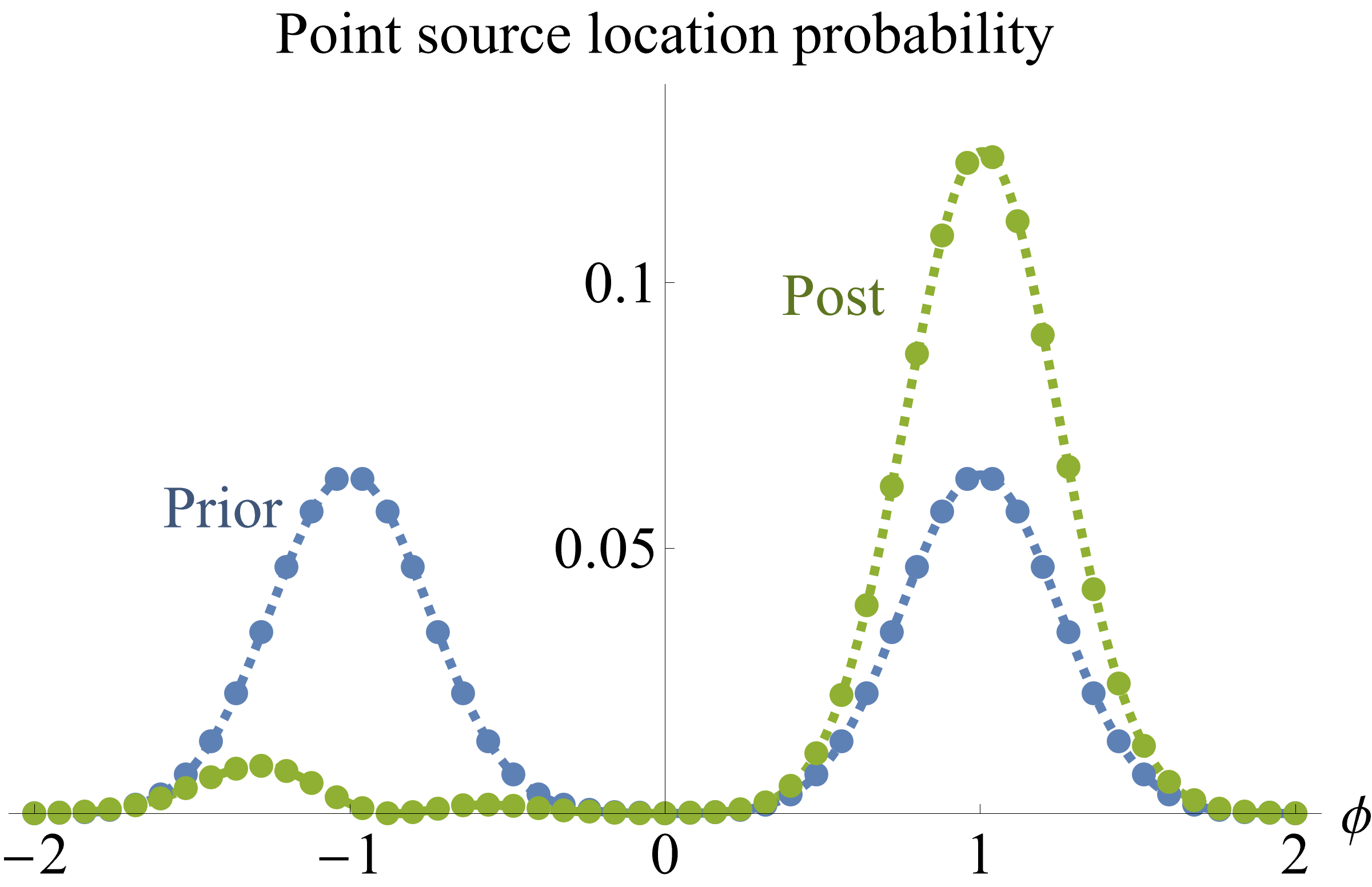}};
        \begin{scope}[x={(imgPDists.south east)},y={(imgPDists.north west)}]
            \node[anchor=north west] at (0,1) {a)};
        \end{scope}
    \end{tikzpicture}%
    \vspace{1em}
  \centering
    \begin{tikzpicture}
        \node[anchor=south west,inner sep=0] (imgCoherences) {\includegraphics[width=\columnwidth]{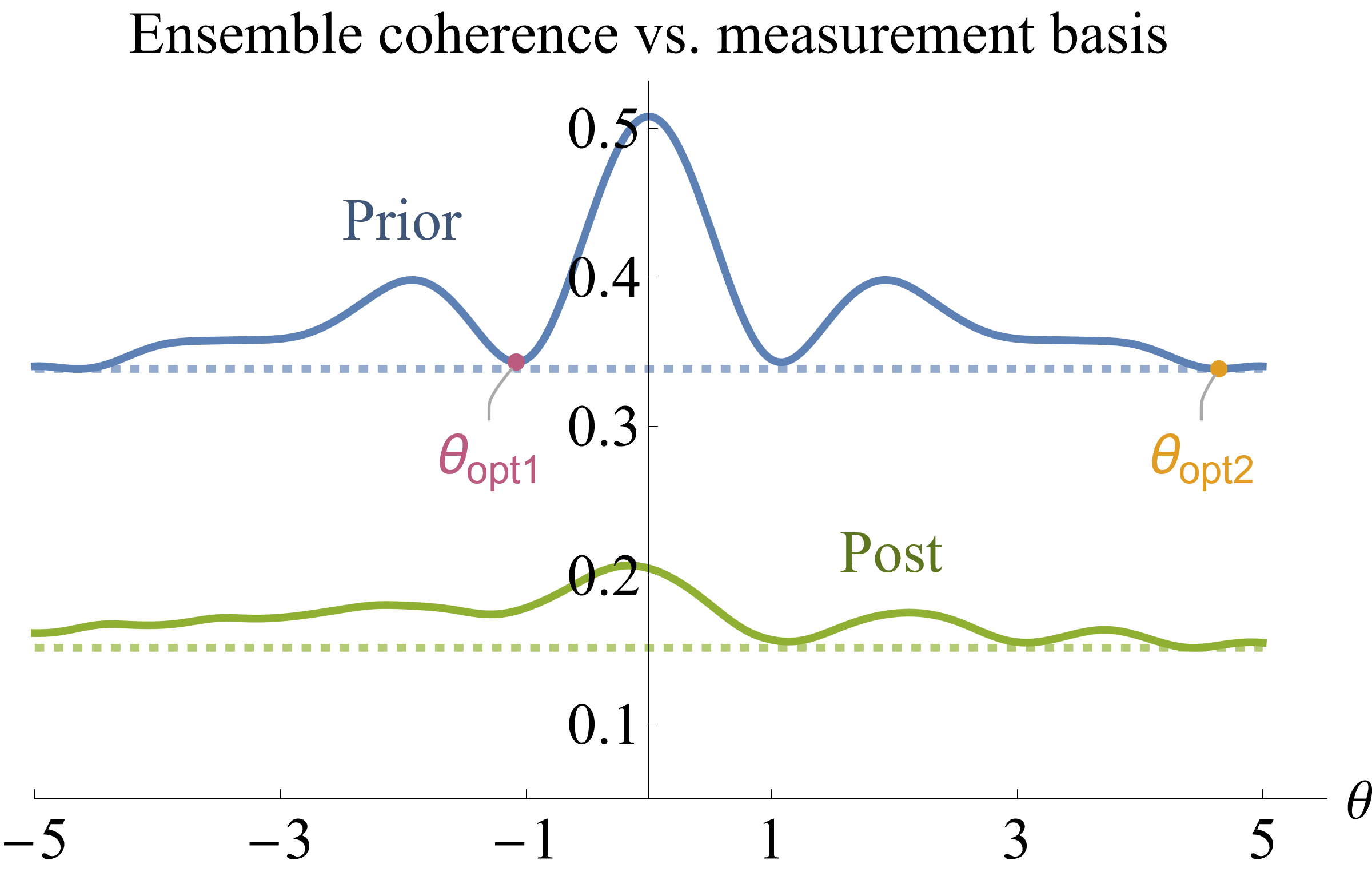}};
        \begin{scope}[x={(imgCoherences.south east)},y={(imgCoherences.north west)}]
            \node[anchor=north west] at (0,1) {b)};
        \end{scope}
    \end{tikzpicture}%
    \vspace{1em}
  \centering
    \begin{tikzpicture}
        \node[anchor=south west,inner sep=0] (imgPerformance) {\includegraphics[width=\columnwidth]{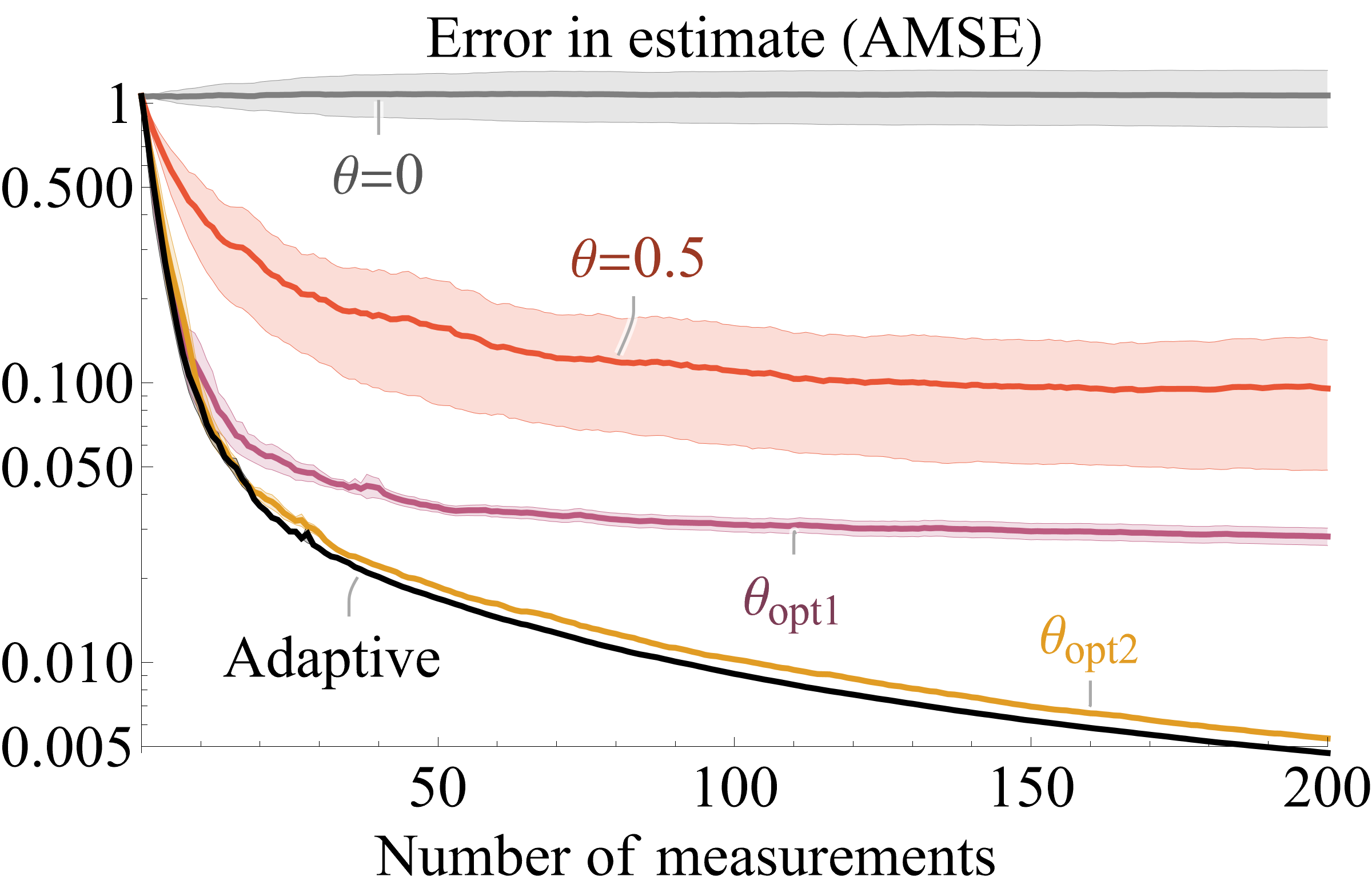}};
        \begin{scope}[x={(imgPerformance.south east)},y={(imgPerformance.north west)}]
            \node[anchor=north west] at (0,1) {c)};
        \end{scope}
    \end{tikzpicture}%
  \caption{a) Probability distributions for the position of the point source. These are discrete, with the dashed lines included for illustration only. `Prior' is the initial distribution. `Post' is the distribution after a single measurement at $\theta_{\mathrm{opt}1}$ yielding $h_2$. b) Ensemble coherences for the same distributions when measured in the Hermite-Gauss basis with displacement $\theta$. The dashed horizontal lines indicate the minimum value. The coherence of the prior distribution at $\theta=0$ is approximately equal to the Holevo information of 0.53, meaning in this measurement basis almost all the information is `locked away' in coherences. c) Performance of measurement with each shift, and how it varies with number of measurements. The solid lines denotes the Average Mean Squared Error (AMSE), and the bands show the corresponding variances. `Adaptive' uses the ensemble coherence to compute the optimal displacement for each measurement, while the other curves repeatedly measure at a constant value of $\theta$.}\label{fig:Adaptive}
\end{figure}

We plot the squared amplitudes of the HG modes in \cref{fig:HGModes}a, as well as photon wavefunctions from sources centered at $\phi=-1$ and $\phi=+1$. We can see that the $\lvert h_q(x)\rvert^2$ are symmetric about $x=0$, thus measurements do not tell us on which side of the origin the source lies. This can be rectified by shifting the centre of the modes, and instead sampling in the basis
\begin{equation}
  \{h_q(x-\theta)\},
\end{equation}
for some displacement $\theta$. Doing this will break the symmetry, so intuitively should provide more information from the same number of measurements. The optimum value of $\theta$ will depend on the prior distribution of $\phi$. Moreover when measuring a series of emitted photons, we should continually adapt $\theta$, incorporating the information from previous measurements. We will find the optimum $\theta$ using the ensemble coherence.

The ensemble coherence quantifies the informational penalty paid due to a choice of measurement basis. Thus minimising coherence will maximise the information gain of the measurement, and consequently minimise the mean squared error. The ensemble coherence depends on the ensemble $\mathcal{E}_{\Phi}$, and hence the probability distribution of $\Phi$. Initially this will be the prior distribution. As we make more measurements, we can use Bayesian probability theory to update the distribution for $\Phi$, and hence find the new optimum displacements. The calculations for this are detailed in \cref{app:AdaptiveMeasurement}.

\cref{fig:Adaptive}a plots the `Prior' distribution $p_{\Phi}(\phi)$ of point source locations before any measurement, and the distribution `Post' incorporating information from a single measurement.
The corresponding ensemble coherences for different shifts $\theta$ are shown in \cref{fig:Adaptive}b. We can see that the prior coherence is symmetric about zero, with local minima at $\theta_{\mathrm{opt1}}\approx\pm1.1$ and $\theta_{\mathrm{opt2}}\approx \pm4.6$. We show the relative position of the photon wavefunctions after these shifts in \cref{fig:HGModes}b. For the post-measurement state the coherence is no longer symmetric, with all optima lying at positive $\theta$.

Once a series of measurements has been taken, we can take the mean of the post-measurement probability distribution to obtain an estimate of $\phi$. This strategy is proven to provide the estimate with the minimum Average Mean-Squared Error (AMSE) \cite[\S 4.2]{trees_detection_2013}. In \cref{fig:Adaptive}c we plot the AMSE after measuring at both constant $\theta$, and with an adaptive scheme that for each measurement picks the shift with minimum ensemble coherence. This plot was generated by randomly simulating 480 sequences of 200 measurements each. The solid lines in \cref{fig:Adaptive}c denote the mean error value. The shaded regions show the variance of the errors, an indicator of how consistent the performance is. Details of these calculations are provided in \cref{app:AdaptiveMeasurement}. 

We can see that $\theta=0$ performs very poorly, and a substantial improvement is gained by breaking the symmetry with $\theta=0.5$. Both of these are outperformed by measuring at $\theta_{\mathrm{opt}1}$, the first local minimum for the Prior distribution as identified by the ensemble coherence. The resulting error is also much more consistent. Substantially better performance is obtained from the second optimum $\theta_{\mathrm{opt}2}$. This may seem surprising, since the difference in coherence between the two shifts is minimal for the Prior distribution. However as we can see in \cref{fig:HGModes}b, $\theta_{\mathrm{opt}2}$ positions the possible source locations so that fringes of the HG modes can unambiguously differentiate between them, and thus is more likely to also be optimal for post-measurement distributions. All of these are outperformed by the Adaptive scheme, which varies $\theta$ for each measurement using the ensemble coherence. We note that this strategy belongs to the class of `greedy' algorithms, which repeatedly maximise short-term gain. 


The constant-shift strategies of $\theta=0,0.5,\theta_{\mathrm{opt}1}$ eventually plateau. Measuring adaptively at the optimum shift however results in a continually decreasing error, with even less variance. Since the adaptive strategy obtains the maximum information from each individual measurement, this curve also provides a bound on the error of any greedy adaptive measurement strategy. While we do not directly observe a plateauing with $\theta_{\mathrm{opt}2}$, an increasing gap opens up between its performance and that of the Adaptive strategy. However the performance is almost identical for a small number of measurements. Thus the coherence allows us to identify both the optimum adaptive scheme, and also non-adaptive schemes whose performance is close to optimal.

Graphing the ensemble coherence in \cref{fig:Adaptive}b also allows us to understand the `measurement landscape'. We can see how sensitive the information gain is to perturbations in the value of $\theta$, and how this changes as we gain more information. Moreover, the value of the ensemble coherence quantifies the difference between information gain and the optimal Holevo information. Since this is non-zero even at the optimum displacements, we know there is advantage to be gained by a more general measurement, using another basis, a POVM, or a multi-photon measurement. Interestingly, the coherence is substantially larger in the `Prior' distribution, likely due to the double-peaked distribution allowing an advantage through some interference scheme. For the posterior distribution the coherence is both lower and flatter, indicating that measurement at any $\theta$ approaches the Holevo information.

\section{Calculations from adaptive measurement example}\label{app:AdaptiveMeasurement}
This section contains details on the example from \cref{sec:Measurement}. The parameter $\phi$ being estimated is the location of the point source. We take this to be discrete for ease of numerical simulation, but results and calculations are analogous for the continuous case. The Prior distribution from \cref{fig:Adaptive}a is constructed by first mixing continuous Gaussians of mean of one-half:
\begin{equation}
  \frac{1}{2}\left(\exp\left(-\frac{(\phi+1)^2}{2(1/2)^2}\right)+\exp\left(-\frac{(\phi-1)^2}{2(1/2)^2}\right)\right),
\end{equation}
selecting fifty points of $\phi$ evenly spaced between $-2$ and $+2$, then normalising the result. This large number of possible locations allows us to approximate a continuous estimation scenario.

The state of a photon emitted from a point source located at $\phi$ is
\begin{equation}
  \ket{\psi_{\phi}}=\int_{-\infty}^{+\infty}\psi_{\phi}(x)\ket{x}\dd x,
\end{equation}
where the wavefunction $\psi_{\phi}$ is defined in \cref{eq:PointSourceWF}. This is measured in the basis of Hermite-Gauss modes. We will denote by $\ket{q,\theta}$ the $q$th mode from \cref{eq:HermiteGaussModes} displaced by the shift $\theta$:
\begin{equation}
  \ket{q,\theta} = \int_{-\infty}^{+\infty} h_q(x-\theta)\ket{x}\dd x.
\end{equation}
A general state can then be expanded as:
\begin{equation}
  \begin{gathered}
  \ket{\psi_{\phi}}=\sum_{q=0}^{\infty} c_q\ket{q,\theta}, \\
  c_q=\int_{-\infty}^{+\infty} h_q(x-\theta)\psi_{\phi}(x)\dd x.
  \end{gathered}
\end{equation}
The corresponding density matrix is then
\begin{equation}
  \rho_{\phi}=\sum_{p,q=0}^{\infty}c_pc_q\ketbra{p,\theta}{q,\theta},
\end{equation}
where all coefficients are real (since the wavefunctions are real). 

The density matrix $\rho_{\phi}$ is the state of a photon emitted by a point source at a known location $\phi$. If $\rho_{\phi}$ is projectively measured in the Hermite-Gauss basis, the measurement probabilities are the diagonal elements of $\rho_{\phi}$ in the $\ket{p,\theta}$ matrix representation. These diagonal elements correspond to the probability distribution of the measurements conditioned on the parameter:
\begin{equation}
  p_{M|\Phi}(m|\phi).
\end{equation}
Averaging this over the possible values of $\phi$ gives the overall measurement probability distribution for an unknown source location:
\begin{equation}
  p_M(m)=\sum_{\phi} p_{M|\Phi}(m|\phi)p_{\Phi}(\phi).
\end{equation}
Putting these together, the updated probability distribution for $\Phi$ incorporating a measurement result follows from Bayes' rule:
\begin{equation}\label{eqa:BayesRule}
  p_{\Phi|M}(\phi|m)=\frac{p_{M|\Phi}(m|\phi)p_{\Phi}(\phi)}{p_M(m)}.
\end{equation}
Thus after each measurement, we can use \cref{eqa:BayesRule} to find the new probability distribution for $\phi$, given our measurement result $m$. This then becomes our $p_{\Phi}(\phi)$, and we can iterate the procedure for each successive measurement.

Let $m$ now denote a sequence of measurement results from the same point source. From this sequence, we must construct an estimate $\hat{\phi}(m)$ of the source location. The estimate minimising the average mean squared error is the mean of the posterior distribution: \cite[\S 4.2]{trees_detection_2013}:
\begin{equation}\label{eq:PhiEstimator}
  \hat{\phi}(m)=\sum_{\phi}\phi\,p_{\Phi|M}(\phi|m),
\end{equation}
where the posterior distribution is computed using Bayes' rule \cref{eqa:BayesRule}. 

In \cref{fig:Adaptive}c we consider five measurement strategies, using both constant and adaptive shifts. The performance of each strategy is quantified using the Average Mean Squared Error (AMSE). This is defined as the average of the squared error over the joint distribution for the parameter and measurement results:
\begin{equation}\label{eq:aAMSE}
  \begin{aligned}
    \mathrm{AMSE} &= \sum_{\phi,m}\left(\hat{\phi}(m)-\phi\right)^2p_{\Phi,M}(\phi,m), \\
      &= \sum_{\phi,m}\left(\hat{\phi}(m)-\phi\right)^2p_{\Phi|M}(\phi|m)p_M(m).
  \end{aligned}
\end{equation}
In other words, the AMSE averages the squared error over all possible source locations $\phi$, and all possible sequences of measurement outcomes $m$. 

Unfortunately, we cannot simulate all possible measurement outcomes. The Hermite-Gauss basis is infinite, so must be truncated at some finite $N_{HG}$. If the source is measured $n$ times, the number of possible measurement sequences grows exponentially as $(N_{HG})^n$. This is too large to consider exactly. Instead, we simulate a random subset of the possible measurement sequences. For each simulated sequence $m$ we compute the average error $e(m)$ as:
\begin{equation}
  e(m)=\sum_{\phi}\left(\hat{\phi}(m)-\phi\right)^2p_{\Phi|M}(\phi|m).
\end{equation}
The AMSE can then be approximated by the mean of these:
\begin{equation}
  AMSE\approx\frac{1}{N}\sum_{m}e(m),
\end{equation}
where the sum is over all simulated measurement sequences. This gives us the solid curves in \cref{fig:Adaptive}c. The bands show the variance of the $e(m)$, as an indicator of how consistent the error is.

It is necessary to check that the number of Hermite-Gauss modes in our basis, and the number of simulated sequences, are sufficient to ensure accurate computation of the AMSE. For the plots in \cref{fig:Adaptive}c, we chose a truncation of $N_{HG}=20$, and simulated $480$ sequences. We confirmed that the results did not change if we increased the truncation to $N_{HG}=30$, or lowered the number of sequences to $320$. The code for this is available for download at \cite{lecamwasam_cxi_2023}. 

Finally, we note that in \cref{fig:Adaptive}b, we graphed the coherence for shifts ranging from $-5\le\theta\le5$. The coherence continues to fall as we go beyond this. This is because as the shift grows larger, we can make use of very fine features in high-order Hermite-Gauss modes to gain increasing amounts of information. However, experimental realisation of sampling in these high modes would be very challenging. Moreover accurate simulation of large shifts requires increasingly large Hilbert spaces for the HG modes. It is thus reasonable to restrict ourselves to the chosen range of shifts. We also note that the adaptive strategy was limited to shifts between $-3\le\theta\le 3$, but increasing the bound from three to five did not noticeably change the performance.


%

\clearpage